\documentclass[journal,10pt]{IEEEtran}
\usepackage[pdftex]{graphicx}
\usepackage{epstopdf}
\usepackage{multirow}
\usepackage{amsmath}
\usepackage{mathtools}
\usepackage{authblk}
\usepackage{cite}
\usepackage{cleveref}
\usepackage{latexsym}
\usepackage[justification=centering]{caption}
\usepackage{amsthm}
\usepackage{balance}
\usepackage{relsize}
\usepackage[nodisplayskipstretch]{setspace}
\usepackage{color}
\usepackage{amssymb}
\usepackage{eucal}
\usepackage{url}
\usepackage{booktabs,subcaption,amsfonts,dcolumn}
\newtheorem{prop}{Proposition}

\usepackage[nodisplayskipstretch]{setspace}
\usepackage{color}
\newcolumntype{d}[1]{D..{#1}}
\usepackage{lipsum}

\usepackage[nodisplayskipstretch]{setspace}
\usepackage{color}
\definecolor{Mypink}{RGB}{255,0,255}
\definecolor{Myorange}{RGB}{255,102,0}
\definecolor{Mygreen}{RGB}{0,153,0}
\definecolor{Myblue}{RGB}{0,0,255}

\DeclareMathAlphabet\mathbfcal{OMS}{cmsy}{b}{n}
\begin{document}

\title{Large Intelligent Surface Assisted Non-Orthogonal Multiple Access: Performance Analysis}

\renewcommand\Authfont{\fontsize{12}{14.4}\selectfont}
\renewcommand\Affilfont{\fontsize{9}{10.8}\itshape}

\author{ Lina~Bariah,~\IEEEmembership{Member,~IEEE,} Sami~Muhaidat,~\IEEEmembership{Senior~Member,~IEEE,} Paschalis C.~Sofotasios,~\IEEEmembership{Senior~Member,~IEEE,} Faissal El Bouanani,~\IEEEmembership{Senior~Member,~IEEE,} Octavia A. Dobre,~\IEEEmembership{Fellow,~IEEE,} and Walaa Hamouda,~\IEEEmembership{Senior~Member,~IEEE,}

				\thanks{L. Bariah and S. Muhaidat are with the KU Center for Cyber-Physical Systems, Department of Electrical and Computer Engineering, Khalifa University, Abu Dhabi, UAE, (e-mails: \{lina.bariah, muhaidat\}@ieee.org.}

				\thanks{P. C. Sofotasios is with the Center for Cyber-Physical Systems, Department of Electrical and Computer Engineering, Khalifa University, Abu Dhabi 127788, UAE, and also with the Department of Electrical Engineering, Tampere University, Tampere 33101, Finland (e-mail: p.sofotasios@ieee.org).}
				
				\thanks{F. El Bouanani is with ENSIAS College of Engineering, Mohammed V University, Rabat, Morocco (e-mail: f.elbouanani@um5s.net.ma).}
				
				\thanks{O. A. Dobre is with the Department of Electrical and Computer Engineer-ing, Memorial University, St. Johns, Canada (e-mail: odobre@mun.ca).}
				
				\thanks{W.  Hamouda  is  with  the  Department  of  Electrical  and Computer Engineering, Concordia University, Montreal, QC, H3G 1M8, Canada (e-mail:hamouda@ece.concordia.ca).}
 
				}
				
\maketitle

\begin{abstract}
\boldmath
Large intelligent surface (LIS) has recently emerged as a potential enabling technology for 6G networks, offering extended coverage and enhanced energy and spectral efficiency. In this work, motivated by its promising potentials, we investigate the error rate performance of LIS-assisted non-orthogonal multiple access (NOMA) networks. Specifically, we consider a downlink NOMA system, in which data transmission between a base station (BS) and $L$ NOMA users is assisted by an LIS comprising $M$ reflective elements. First, we derive the probability density function of the end-to-end wireless fading channels between the BS and NOMA users. Then, by leveraging the obtained results, we derive an approximate expression for the pairwise error probability (PEP) of NOMA users under the assumption of imperfect successive interference cancellation. Furthermore, accurate expressions for the PEP for $M = 1$ and large $M$ values ($M > 10$) are presented in closed-form. To gain further insights into the system performance, an asymptotic expression for the PEP in high signal-to-noise ratio regime, the achievable diversity order, and a tight union bound on the bit error rate are provided. Finally, numerical and simulation results are presented to validate the derived mathematical results.
\end{abstract}
 
\begin{keywords}
Diversity order, large intelligent surfaces (LIS), Meijer's G-function, non-orthogonal multiple access (NOMA), pairwise error probability.
\end{keywords}

\IEEEpeerreviewmaketitle
\section{Introduction}

The ongoing deployment of the fifth-generation (5G) mobile networks aims to realize the vision of ``connected everything,'' enabling distinguished paradigms of wireless communications, including machine-to-people communications, machine-to-machine communications, and people-to-people communications \cite{marco}. These new paradigms are expected to connect a massive number of heterogeneous data-hungry devices with diverse requirements, resulting in a spectrum shortage crisis, in addition to degraded latency and increased energy consumption. This stimulates the call for the development of innovative technologies, which are capable of supporting massive connectivity, extremely high data rates and ultra-low latency.

A promising technology, i.e., non-orthogonal multiple access (NOMA), has emerged as a key enabling technology for 5G networks, which offers reduced latency, enhanced connectivity and reliability, and improved energy and spectral efficiency \cite{Barrieh,Barrieh3}. The principle concept of downlink (DL) NOMA is to employ superposition coding at the base station (BS) by sorting out users based on their channel gains, and then, multiplex their signals in power domain \cite{8357810,7676258}. Interference mitigation in NOMA is carried out by performing successive interference cancellation (SIC) for users' signals with higher power levels \cite{Barrieh2}. On the other hand, users' signals with low power levels are treated as noise. To realize efficient SIC, and hence, efficient NOMA system, paired users should have distinct channel conditions \cite{8823873}. 

Due to the highly stochastic nature of the wireless channels, wireless incident signals are susceptible to several impairments, such as random fluctuation, path loss, blockage, and absorption, yielding uncontrollable multipath fading environments. Such impairments can be overcome by employing multiple-input multiple-output (MIMO) schemes. However, this comes at the expense of reduced energy efficiency, which constitutes a fundamental challenge in the 5G and beyond (5GB) wireless networks \cite{Sousa}. Large intelligent surface (LIS) has recently emerged as a disruptive energy and spectrally efficient technology, which is capable of offering a programmable control over the wireless environment \cite{8796365,6G}. This can be realized by incorporating reflective elements (REs) which manipulate the impinging electromagnetic waves to perform various functionalities, such as wave reflection, refraction, absorption, steering, focusing and polarization \cite{8449754}.

In the context of NOMA, it was shown that LIS can tune the propagation environment to guarantee specific users' order, and hence, enabling efficient deployment of NOMA systems. Furthermore, by efficiently utilizing LIS, it becomes more feasible to change users' order to satisfy a particular system requirement according to certain users priorities, rather than relying on the random propagation environment of wireless systems \cite{Sousa}. Therefore, motivated by the envisioned potentials of LIS, the integration of LIS into NOMA systems has recently received an increasing attention from the research community \cite{8970580,9000593,9024675,2019arXiv191013636M,2019arXiv190703133Y,2019arXiv191007361F}.

\subsection{Related work}

Existing research works have primarily focused on investigating the performance of LIS-assisted NOMA systems, in addition to proposing enhancement schemes to optimize the performance of these systems in terms of outage probability and achievable rate. In \cite{2019arXiv191007361F}, the authors formulated a joint optimization problem to obtain an optimum beamforming vector and phase-shift matrix at the LIS that minimize the total transmission power of a downlink NOMA system. It was demonstrated that LIS can potentially reduce the total transmission power in NOMA systems. The authors in \cite{2019arXiv190906972L} considered an LIS-assisted multi-cluster multiple-input single-output (MISO) NOMA system. They proposed a mathematical framework in order to jointly optimize the beamforming vector at BS and the reflection coefficient vector at LIS. In \cite{2019arXiv190703133Y}, the authors formulated an optimization problem and proposed a solution to maximize users' data rates, while ensuring users' fairness in a downlink LIS-based NOMA system. Transmit beamforming and phase-shift vector optimization have been investigated in the literature under different scenarios and for various setups \cite{8970580,2019arXiv191013636M,9000593,2020arXiv200302117H, Octavia,Octavia2}. The authors in \cite{2019arXiv191210044H} analyzed outage probability, ergodic rate, diversity order, and spectral and energy efficiencies of users with higher priority in an LIS-assisted NOMA scenario while considering direct links between the BS and NOMA users. The reported results in \cite{2019arXiv191210044H} demonstrated that the effect of the direct link is negligible, especially for a large number of REs. In \cite{2020arXiv200209907Y}, the authors investigated the effect of residual interference due to imperfect SIC on the outage probability and ergodic rate performance of LIS-based NOMA systems. The reported results in \cite{2020arXiv200209907Y} highlighted the advantages of incorporating NOMA with LIS, compared to OMA, to enhance the outage probability and ergodic rate performance.

While the performance of LIS-assisted NOMA systems has been analyzed in the literature, a comprehensive error rate analysis is still missing, although it is the most revealing metric about the system performance. The error rate performance of an LIS-based NOMA system was addressed in \cite{2020arXiv200209453C} for the special case of two users. The authors in \cite{2020arXiv200209453C} derived two bit error rate (BER) expressions for the two users under specific modulation schemes. Moreover, the central limit theorem was adopted in \cite{2020arXiv200209453C}, limiting the analytical framework to a large number of REs.

\subsection{Contribution}

Motivated by the above discussion, in this paper, we provide a comprehensive analytical framework to investigate the error rate performance of an LIS-based NOMA system. The proposed work is valid for an arbitrary modulation scheme and any number of users. The main contributions of this paper are summarized as follows:
\begin{itemize}
\item We derive a tight approximate probability density function (PDF) expression for the end-to-end (e2e) wireless channel of the underlying downlink LIS-based NOMA system, with $L$ users and $M$ REs.
\item We derive a novel pairwise error probability (PEP) expression of NOMA users, under the assumption of imperfect SIC. The derived PEP expression is then utilized to obtain a union bound on the error rate.
\item We derive PEP expressions for the special cases of $M = 1$ and large $M$ ($M > 10$), where the latter is obtained by utilizing the central limit theorem.
\item We further analyze the asymptotic PEP for general $M$ values. The derived asymptotic PEP is then used to investigate the achievable diversity order and quantify the effect of the number of REs on the error rate performance of NOMA users.
\end{itemize} 

The rest of the paper is organized as follows. Section \ref{sec:model} introduces the system model of a downlink LIS-based NOMA system. An accurate PEP analysis for an arbitrary number of REs ($M$) is presented in Section \ref{sec:pep}, while the asymptotic PEP analysis and the achievable diversity order are investigated in Section \ref{sec:asy}. Numerical and simulation results are presented in Section \ref{sec:result} and the paper is concluded in Section \ref{sec:con}.

\textit{Notations:} The complex conjugate operation and the absolute value are denoted as $(\cdot )^{\ast}$ and $\left |. \right |$, respectively. $\Re \left\{.\right\}$ represents the real part of a complex number, $\hat{x}$ is a detected symbol, while $\bar{x}$ denotes an erroneously detected symbol. The expectation operator is denoted as $\mathbb{E}\left ( . \right )$. $\CMcal{CN}(\mu,\sigma^2)$ represents the circularly symmetric complex Gaussian (CSCG) distribution with mean $\mu$ and variance $\sigma^2$.

\section{System Model} 
\label{sec:model}

In this work, we consider a downlink LIS-assisted NOMA system comprising a single LIS with $M$ REs to assist in data transmission between a single BS and $L$ NOMA users, $U_{1},\cdots,U_{L}$, as depicted in Fig. \ref{fig:model}. Following the principle of NOMA, the BS broadcasts a superposed message that consists of $L$ messages intended for the $L$ users, multiplexed in power domain. In our work, without loss of generality, we consider that the first user, $U_{1}$, is the farthest user from the LIS, i.e., the user with the weakest link. On the other hand, the $L$th user, $U_{L}$, is the nearest user to the LIS, and therefore, it experiences the strongest channel, i.e., $d_{R,1}>\cdots>d_{R,L}$, where $d_{R,l}$ represents the distance from the LIS to the $l$th user. The power coefficients are allocated to the users based on their locations, in which far users are allocated higher power levels than near users.
\begin{figure}[ht] 
\centering
\includegraphics[width=0.7\linewidth]{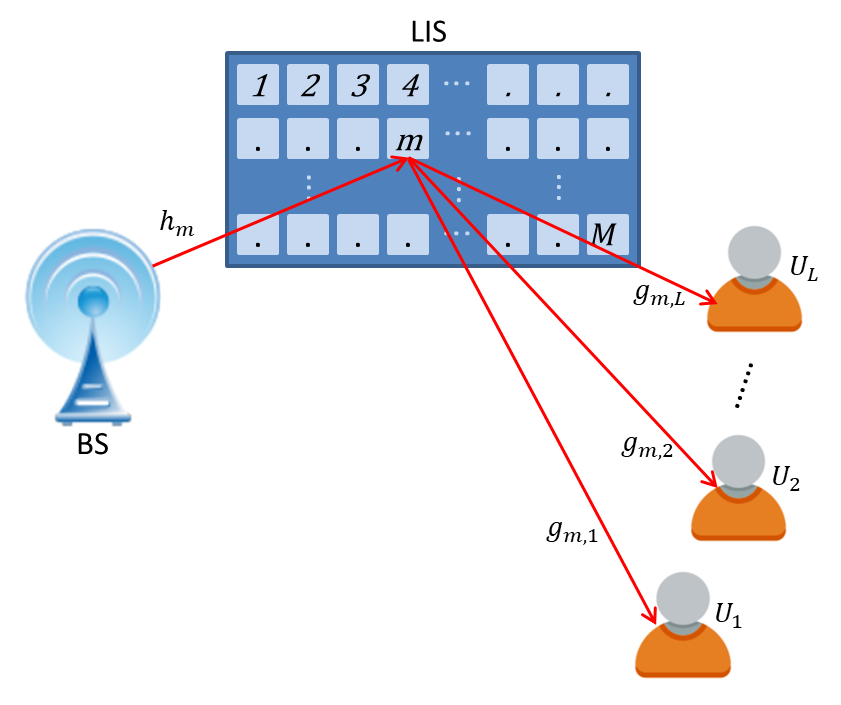}
\caption{LIS-based NOMA system model with $L$ users.}
\label{fig:model}
\end{figure}

The baseband received signal at the $l$th user can be written as \vspace{-0.2cm}
\begin{equation}\vspace{-0.2cm}
\label{eq:r}
r_{l} = \frac{1}{\sqrt{d_{B}^{\alpha}d_{R,l}^{\alpha}}}\sum_{m=1}^{M}h_{m}g_{m,l}w_{m}\textup{e}^{j\theta_{m}}\sum_{i=1}^{L}\sqrt{P_{i}}\,x_{i}+n_{l},
\end{equation} 
where $x_{i}$ and $P_{i}$ are the transmitted symbol and the power coefficient of the $i$th user, respectively, and $\alpha$ accounts for the path-loss exponent. For simplicity, the total transmission power of the BS is normalized to unity. Note that the distance between the BS and the LIS is denoted by $d_{B}$. Moreover, $h_{m}$ represents the small-scale fading channel between the BS and the $m$th RE, while $g_{m,l}$ represents the fading channel between the $m$th RE of the LIS and the $l$th user. Both $h_{m}$ and $g_{m,l}$ are CSCG distributed with zero mean and variance $\sigma^{2}$. The additive white Gaussian noise (AWGN) is represented by $n_{l} \sim \CMcal{CN}\left ( 0,N_{0} \right )$. Additionally, $w_{m}$ represents the amplitude reflection coefficient, which is assumed to be normalized to unity \cite{2019arXiv191013636M}\cite{2019arXiv191211768Z}, while $\theta_{m}$ is the phase shift applied by the $m$th RE. Assuming perfect knowledge of the channels phases, $\theta_{h_m}$ and $\theta_{g_{m,l}}$ of $h_m$ and $g_{m,l}$, respectively, the phase shift of the $m$th RE is chosen as $\theta_{m} = -(\theta_{h_m}+\theta_{g_{m,l}})$. Subsequently, the received signal at the $l$th user can be rewritten as \vspace{-0.2cm}
\begin{equation} \vspace{-0.2cm}
\label{eq:r2}
r_{l} = q_{l}\sum_{i=1}^{L}\sqrt{P_{i}}\,x_{i}+n_{l},
\end{equation} 
where $q_{l}$ denotes the e2e channel coefficient, which is given by \vspace{-0.2cm}
\begin{equation} \vspace{-0.2cm}
\label{eq:q}
q_{l}=\frac{1}{\sqrt{d_{B}^{\alpha}d_{R,l}^{\alpha}}}\sum_{m=1}^{M}\left |h_{m}\right | \left |g_{m,l}\right |.
\end{equation}
Multi-user interference at the $l$th user is canceled by employing SIC for higher power signals, i.e., $x_{1},\cdots,x_{l-1}$, whereas the lower power users' signals, i.e., $x_{l+1},\cdots,x_{L}$, are treated as additive noise. Hence, the output signal of the $l$th SIC receiver can be written as \vspace{-0.2cm}
\begin{equation} \vspace{-0.2cm}
\label{eq:r3}
r'_{l} = q_{l}\left (\sqrt{P_{l}}\,x_{l}+X_{l}  \right )+n_{l},
\end{equation}  
with \vspace{-0.2cm}
\begin{equation} \vspace{-0.2cm}
\label{eq:X}
X_l=\sum_{i=1}^{l-1}\sqrt{P_{i}}\,\delta_{i}+\sum_{j=l+1}^{L}\sqrt{P_{j}}\,x_{j},
\end{equation}  
where $\delta_{i}=x_{i}-\hat{x}_{i}$ denotes the output of the $i$th SIC iteration and $\hat{x}_{i}$ denotes the corresponding detected symbol. Note that $\delta_{i}$ has two scenarios, namely $\delta_{i}=0$ and $\delta_{i} \neq 0$ for successful and unsuccessful SIC, respectively. 
\section{Pairwise Error Probability Analysis} 
\label{sec:pep}
In this section, we derive an accurate expression for the PEP, which offers useful insights into the error rate performance of the considered scenario. Moreover, the derivation of the PEP constitutes the basic block for obtaining an upper bound on the BER. The PEP of the $l$th user is defined as the probability of erroneously decoding symbol $\bar{x}_{l}$ given that symbol $x_{l}$ was transmitted, where $\bar{x}_{l} \neq x_{l}$, $\forall l$. Therefore, based on the maximum-likelihood (ML) rule, the conditional PEP of the $l$th user can be expressed as \vspace{-0.2cm}
\begin{equation} \vspace{-0.2cm}
\label{eq:pep}
\begin{split}
\textup{Pr}\left ( x_{l}\rightarrow \bar{x}_{l} \mid q_{l} \right )= \textup{Pr}\Bigg (   \left | r'_{l}-q_{l}\sqrt{P_{l}}\bar{x}_{l} \right |^{2} \leq \left | r'_{l}-q_{l}\sqrt{P_{l}}x_{l} \right |^{2} \Bigg ).
\end{split}
\end{equation}  
By inserting \eqref{eq:r3} into \eqref{eq:pep}, and after some mathematical simplifications, the PEP can be rewritten as \vspace{-0.2cm}
\begin{equation} \vspace{-0.2cm}
\label{eq:pep1}
\begin{split}
\textup{Pr}\left ( x_{l}\rightarrow \bar{x}_{l} \mid q_{l} \right )= \textup{Pr}\Bigg ( &\left | q_{l}(\sqrt{P_{l}}\bar{\Delta}_{l}+X_{l})+n_{l} \right |^{2} \\
& \;\;\;\;\;\;\;\;\;\;\;\;\;\;\;\;\;\; \leq \left | q_{l}X_{l}+n_{l} \right |^{2} \Bigg ),
\end{split}
\end{equation} 
where $\bar{\Delta}_{l}=x_{l}-\bar{x}_{l}$. 

After expanding \eqref{eq:pep1}, the conditional PEP can be expressed as in \eqref{eq:pep2}, at the top of next page.
\begin{figure*}[ht]
\begin{equation}
\label{eq:pep2}
\textup{Pr}\left ( x_{l}\rightarrow \bar{x}_{l} \mid q_{l} \right )= \textup{Pr}\left ( \underbrace{2\sqrt{P_{l}}\Re \left \{ \bar{\Delta}_{l} n_{l}^{\ast}\right \} }_{N_l}\leq -q_{l}\left (\left |\sqrt{P_{l}}\bar{\Delta}_{l}+X_{l}  \right |^{2}-\left | X_{l} \right |^{2}  \right )  \right ).
\end{equation} 
\hrulefill
\vspace*{4pt}
\end{figure*}
Conditioned on $q_{l}$, the noise term $N_l$ in \eqref{eq:pep2} is modeled as a Gaussian random variable (RV) with zero mean and variance $\sigma_{N}^{2} = 2P_{l}\left |\bar{\Delta}_{l}  \right |^{2} N_{0}$. Hence, the conditional PEP can be evaluated as \vspace{-0.2cm}
\begin{equation} \vspace{-0.2cm}
\label{eq:pep3}
\textup{Pr}\left ( x_{l}\rightarrow \bar{x}_{l} \mid q_{l}=q \right ) = Q\left (  \frac{{q}\vartheta_{l} }{\lambda_{l} } \right ),
\end{equation} 
where $Q(\cdot)$ represents the Gaussian $Q$-function. Also,\vspace{-0.2cm} 
\begin{equation}\vspace{-0.2cm}
\label{eq:v}
\vartheta_{l}=\sqrt{P_{l}}\left |\bar{\Delta}_{l}  \right |^{2}+2\Re \left \{ \bar{\Delta}_{l} X_{l}^{\ast} \right \}
\end{equation}
and\vspace{-0.2cm}
\begin{equation}\vspace{-0.2cm}
\label{eq:lambda}
\lambda_{l}=\left |\bar{\Delta}_{l}  \right |\sqrt{2N_{0}}.
\end{equation}
By applying the Chernoff bound on the $Q$-function in \eqref{eq:pep3}, the conditional PEP can be bounded by \cite{CB}
\begin{equation}
\label{eq:pep4}
\textup{Pr}\left ( x_{l}\rightarrow \bar{x}_{l} \mid q_{l} = q\right ) \leq \textup{exp}\left ( - \frac{q^{2}\vartheta_{l}^{2} }{2\lambda_{l}^{2} } \right ).
\end{equation} 
To obtain the unconditional PEP, we average over the PDF of the total fading coefficient, $q_{l}$.
\begin{prop}
\label{lemma:pdf}
The PDF of $q_{l}$ can be approximated by
\newcounter{tempequationcounter}
 \begin{figure*}[t]
\setcounter{equation}{23}
\begin{equation}\label{m3a}
\mu_{3}=\left\{ \begin{array}{l}
M\pi\left(\frac{9}{2}+6(M-1)+(M-1)(M-2)\frac{\pi^{2}}{8}\right)\sigma^{6},M\geq3\\
21\pi\sigma^{6},M=2\\
\frac{9\pi\sigma^{6}}{2},M=1
\end{array}\right.
\end{equation}
\hrulefill
\vspace*{4pt}
\end{figure*}
\setcounter{equation}{12}
\begin{equation}\label{ccdf1}
f_{q_{l}}(q)\approx a_{1}\sqrt{d_{B}^{\alpha}d_{R,l}^{\alpha}} G_{1,2}^{2,0}\left(\frac{\sqrt{d_{B}^{\alpha}d_{R,l}^{\alpha}} q}{a_{2}} \Bigg| \begin{array}{c}
-;a_{3}\\
a_{4},a_{5};-
\end{array} \right), 
\end{equation} \noindent
\begin{figure*}[ht] 
\setcounter{equation}{24}
\begin{align}\label{mu4}
\mu_{4}=\left\{ \begin{array}{l}
\left(64M+48M(M-1)+9M(M-1)\pi^{2}+6M(M-1)(M-2)\pi^{2}+\frac{M(M-1)(M-2)(M-3)\pi^{4}}{16}\right)\sigma^{8},M\geq4\\
\left(480+90\pi^{2}\right)\sigma^{8},M=3\\
\left(224+18\pi^{2}\right)\sigma^{8},M=2\\
64\sigma^{8},M=1
\end{array}\right.
\end{align}
\hrulefill
\vspace*{4pt}
\end{figure*}
\setcounter{equation}{13}
\noindent where $G^{m,n}_{p,q}\left ( .\mid. \right )$ denotes the standard Meijer's G-function,
\begin{equation}
a_{1}=\frac{\Gamma(a_{3}+1)}{a_{2}\Gamma(a_{4}+1)\Gamma(a_{5}+1)}, \label{a1}
\end{equation}
\begin{equation}
a_{2}=\frac{a_{3}}{2}\left(\varphi_{4}-2\varphi_{3}+\varphi_{2}\right)+2\varphi_{4}-3\varphi_{3}+\varphi_{2}, \label{a2}
\end{equation}
\begin{equation}
a_{3}=\frac{4\varphi_{4}-9\varphi_{3}+6\varphi_{2}-\mu_{1}}{-\varphi_{4}+3\varphi_{3}-3\varphi_{2}+\mu_{1}}, \label{a3}
\end{equation}
\begin{equation}
a_{4}=\frac{a_{6}+a_{7}}{2}, \label{a4}
\end{equation}
\begin{equation} 
a_{5}=\frac{a_{6}-a_{7}}{2}, \label{a5}
\end{equation}
\begin{equation}
a_{6}=\frac{a_{3}\left(\varphi_{2}-\mu_{1}\right)+2\varphi_{2}-\mu_{1}}{a_{2}}-3,
\end{equation}
\begin{equation}
a_{7}=\sqrt{\left(\frac{a_{3}\left(\varphi_{2}-\mu_{1}\right)+2\varphi_{2}-\mu_{1}}{a_{2}}-1\right)^{2}-4\frac{\mu_{1}(a_{3}+1)}{a_{2}}},
\end{equation}
and
\begin{equation}
\varphi_{i}=\frac{\mu_{i}}{\mu_{i-1}},i \geq 1,
\end{equation}
where $\Gamma(\cdot)$ represents the complete Gamma function \cite{IntTable}, and $\mu_{i}$ is the $i$th moment of $q_{l}$ and $\mu_{0}=1$. The first four moments of $q_{l}$ are given by
\begin{equation}\label{mu1}
\mu_{1} =\frac{M\pi\sigma^{2}}{2},
\end{equation}
\begin{equation}\label{m2a}
\mu_{2}  =\left(4+(M-1)\frac{\pi^{2}}{4}\right)M\sigma^{4},
\end{equation}
and, $\mu_3$ and $\mu_{4}$ are given in \eqref{m3a} and \eqref{mu4}, respectively, at the top of this page.
\stepcounter{equation}
\stepcounter{equation}
\end{prop}
\begin{proof}
The derivation of the PDF and corresponding moments are provided in Appendix I.
\end{proof}
Therefore, using the derived PDF jointly with \eqref{eq:pep4}, the unconditional PEP can be tightly approximated by \eqref{eq:pep5}, which is given at the top of next page.
\begin{figure*}[ht]
\begin{equation}
\label{eq:pep5}
\textup{Pr}\left ( x_{l}\rightarrow \bar{x}_{l} \right ) \approx a_{1}\sqrt{d_{B}^{\alpha}d_{R,l}^{\alpha}}\int_{0}^{\infty}\textup{exp}\left ( - \frac{q^{2}\vartheta_{l}^{2} }{2\lambda_{l}^{2} } \right ) G_{1,2}^{2,0}\left(\frac{\sqrt{d_{B}^{\alpha}d_{R,l}^{\alpha}} q}{a_{2}} \Bigg| \begin{array}{c}
-;a_{3}\\
a_{4},a_{5};-
\end{array} \right)\, dq. 
\end{equation} 
\hrulefill
\vspace*{4pt}
\end{figure*}
Utilizing the following Meijer's G-function representation of the exponential function \cite[Eq. (8.4.3.1)]{int22}
\begin{equation}
\label{eq:exp}
\textup{exp}(-z)=G^{1,0}_{0,1}\left(z\left\vert \begin{array}{c}
-;-\\
0;-
\end{array}\right.\right),
\end{equation} 
the integral in \eqref{eq:pep5} can be rewritten as in \eqref{eq:pep6}, at the top of next page.
\begin{figure*}[ht]
\begin{equation}
\label{eq:pep6}
\textup{Pr}\left ( x_{l}\rightarrow \bar{x}_{l} \right ) \approx a_{1}\sqrt{d_{B}^{\alpha}d_{R,l}^{\alpha}}\int_{0}^{\infty}G^{1,0}_{0,1}\left(\frac{q^{2}\vartheta_{l}^{2} }{2\lambda_{l}^{2} }\left\vert \begin{array}{c}
-;-\\
0;-
\end{array}\right.\right) G_{1,2}^{2,0}\left(\frac{\sqrt{d_{B}^{\alpha}d_{R,l}^{\alpha}} q}{a_{2}} \Bigg| \begin{array}{c}
-;a_{3}\\
a_{4},a_{5};-
\end{array}\right)\, dq. 
\end{equation} 
\hrulefill
\vspace*{4pt}
\end{figure*}
The integral in \eqref{eq:pep6} can be solved using \cite[Eq. (2.24.1.1)]{int22}. Hence, the unconditional PEP can be evaluated as in \eqref{eq:pep7}, given at the top of next page, where
\begin{figure*}[ht]
\begin{equation}
\label{eq:pep7}
\textup{Pr}\left ( x_{l}\rightarrow \bar{x}_{l} \right ) \approx \frac{a_{1}a_{2}2^{a_{4}+a_{5}-a_{3}}}{\sqrt{\pi}}\, G^{1,4}_{4,3}\left(\zeta_{l} \Bigg| \begin{array}{c}
\frac{-a_{4}}{2},\frac{1-a_{4}}{2},\frac{-a_{5}}{2},\frac{1-a_{5}}{2};-\\
0;\frac{-a_{3}}{2},\frac{1-a_{3}}{2}
\end{array}\right).
\end{equation} 
\hrulefill
\vspace*{4pt}
\end{figure*}
\begin{equation}
\zeta_{l} = \frac{2a_{2}^{2}\vartheta_{l}^{2}}{d_{B}^{\alpha}d_{R,l}^{\alpha}\lambda_{l}^{2}}.
\end{equation}
In the following, to obtain insights into the system performance, we consider the following two special cases, in which the corresponding PEP expressions are derived. 

\noindent \textbf{Case 1 ($M = 1$):}
For $M=1$, the PEP of the $l$th user can be upper bounded by
\begin{equation}
\label{eq:pepm2}
\textup{Pr}\left ( x_{l}\rightarrow \bar{x}_{l} \right ) \leq \eta_l \, \textup{exp} \left ( \eta_l \right )\textup{E}_{1}\left ( \eta_l \right ),
\end{equation} 
where $\textup{E}_{1}\left ( z \right )=\int_{z}^{\infty}(e^{-t}/t)\, dt $ represents the exponential integral function, and
\begin{equation}
\label{eq:eta}
\eta_l = \frac{\lambda_{l}^{2}d_{B}^{\alpha}d_{R,l}^{\alpha}}{2\sigma^{4}\vartheta_{l}^{2}}.
\end{equation} 
\begin{proof}
For the special case of $M=1$, the fading coefficient $q_{l}$ reduces to 
\begin{equation}
\label{eq:q1}
(q_{l})_{M=1} \rightarrow \grave{q}_{l}=\frac{\left |h_{1}\right | \left |g_{1,l}\right |}{\sqrt{d_{B}^{\alpha}d_{R,l}^{\alpha}}},
\end{equation} 
which is a double Rayleigh RV with PDF given by \cite{DR} 
\begin{equation}
f_{\grave{q}_{l}}\left(q\right)=\frac{q\sqrt{d_{B}^{\alpha}d_{R,l}^{\alpha}}}{\sigma^{4}}K_{0}\left(\frac{q\sqrt{d_{B}^{\alpha}d_{R,l}^{\alpha}}}{\sigma^{2}}\right), \label{eq:PDFDR1}
\end{equation}
where $K_{0}(\cdot)$ represents the zero-order modified Bessel function of the second kind. Subsequently, using \eqref{eq:pep4} and \eqref{eq:PDFDR1}, the unconditional PEP for $M=1$ can be evaluated as follows
\begin{equation}
\label{eq:pepm1}
\begin{split}
\textup{Pr}\left ( x_{l}\rightarrow \bar{x}_{l} \right ) \leq & \frac{\sqrt{d_{B}^{\alpha}d_{R,l}^{\alpha}}}{\sigma^4}\int_{0}^{\infty}q\,\textup{exp}\left ( - \frac{q^{2}\vartheta_{l}^{2} }{2\lambda_{l}^{2} } \right )\\
& \times K_{0}\left(\frac{q\sqrt{d_{B}^{\alpha}d_{R,l}^{\alpha}}}{\sigma^{2}}\right)  dq. 
\end{split}
\end{equation} 
Finally, the integral in \eqref{eq:pepm1} can be solved using \cite[Eq. (2.16.8.6)]{int2}, yielding \eqref{eq:pepm2}.
\end{proof}

\noindent \textbf{Case 2 ($M > 10$):}
For a sufficiently large number of REs, the PEP of the $l$th user can be evaluated as
\begin{equation}
\label{eq:pepm4}
\begin{split}
\textup{Pr}\left ( x_{l}\rightarrow \bar{x}_{l} \right ) \leq \sqrt{\frac{d_{B}^{\alpha}d_{R,l}^{\alpha}}{8\xi_{l}\bar{\sigma}_{l}^{2}}}  &\, \textup{exp}\left (  \frac{d_{B}^{\alpha}d_{R,l}^{\alpha}\bar{\mu}_{l}^{2}}{4\xi_{l}\bar{\sigma}_{l}^{4}} -\frac{\bar{\mu}_{l}^{2}}{2\bar{\sigma}_{l}^{2}} \right ) \\
& \times \left [1+\textup{erf}\left ( \sqrt{\frac{d_{B}^{\alpha}d_{R,l}^{\alpha}\bar{\mu}_{l}^{2}}{4\xi_{l}\bar{\sigma}_{l}^{4}}} \right )  \right ], 
\end{split}
\end{equation} 
where
\begin{equation}
\label{eq:xi}
\xi_{l} = \frac{\vartheta_{l}^{2}\bar{\sigma}_{l}^{2}+d_{B}^{\alpha}d_{R,l}^{\alpha}\lambda_{l}^{2} }{2\bar{\sigma}_{l}^{2}\lambda_{l}^{2}},
\end{equation}
\begin{equation}\label{nmu}
\bar{\mu}_{l}= \frac{M\pi \sigma^{2}}{2},
\end{equation} 
\begin{equation}\label{nvar}
\bar{\sigma}_{l}^{2}= \frac{M\sigma^{2}\left ( 16-\pi^2 \right )}{4},
\end{equation} 
and $\textup{erf}\left ( \cdot \right )$ denotes the error function.
\begin{proof}
For a sufficiently large number of REs, $M > 10$, according to the central limit theorem \cite{probability}, the e2e cascaded fading coefficient $q_{l}$ in \eqref{eq:q} reduces to a normally distributed RV, $(q_{l})_{M > 10} \rightarrow \tilde{q}_{l} \sim \CMcal{N}\left ( \bar{\mu}_{l},\bar{\sigma}_{l} \right )$ with mean and variance obtained from \eqref{nmu} and \eqref{nvar}, respectively. 

Hence, the PDF of $\tilde{q}_{l}$ can be expressed as the following
\begin{equation} \label{eq:barpdf}
f_{\tilde{q}_{l}}(q)= \frac{1}{\sqrt{2 \pi \bar{\sigma}_{l}^{2}}} \, \textup{exp}\left ( -\frac{\left ( q-\bar{\mu}_{l} \right )^{2}}{2 \bar{\sigma}_{l}^2} \right ).
\end{equation} 
Accordingly, for the case of a large number of REs, the PEP can be evaluated by averaging \eqref{eq:pep4} over $\tilde{q}$, which has the PDF given by \eqref{eq:barpdf}. After some mathematical manipulations, the unconditional PEP can be written as follows
\begin{equation}
\label{eq:pepm3}
\begin{split}
\textup{Pr}\left ( x_{l}\rightarrow \bar{x}_{l} \right ) \leq \frac{\textup{e}^{-\frac{\bar{\mu}_{l}^{2}}{2\bar{\sigma}_{l}^{2}}}}{\sqrt{2 \pi \bar{\sigma}_{l}^{2}}}\int_{0}^{\infty}\textup{e}^{\left ( - \frac{q^{2}\left (\vartheta_{l}^{2}\bar{\sigma}_{l}^{2}+\lambda_{l}^{2} \right )}{2\bar{\sigma}_{l}^{2}\lambda_{l}^{2} }+\frac{ \bar{\mu}_{l}q }{\bar{\sigma}_{l}^2} \right )}   dq.  
\end{split}
\end{equation} 
The integral in \eqref{eq:pepm3} can be evaluated using \cite[Eq. (3.322.2)]{IntTable}. That is, \eqref{eq:pepm4} is obtained, which concludes the proof.
\end{proof}
\noindent \underline{\textbf{Union Bound on the BER Performance:}}

The derived PEP in \eqref{eq:pep7} is leveraged to compute an accurate upper bound on the average error probability of NOMA users. Note that the BER union bound is defined as the weighted sum of all PEP values, considering all transmitted and detected symbols scenarios for all users, including perfect and imperfect SIC \cite{tellambura}. Hence, utilizing the derived PEP expression in \eqref{eq:pep7}, the BER union bound can be written as \cite[Eq. (4.2–70)]{proakis}
\begin{equation}
P_{e}=\frac{1}{\tau}\sum_{x_{l}}^{ }\sum_{\underset{\bar{x}_{l} \neq x_{l}}{\bar{x}_{l}} }^{ }\textup{Pr}\left ( x_{l}\rightarrow \bar{x}_{l} \right ),
\end{equation}
where $\tau$ denotes the number of possible combinations of $x_{l}$ and $\bar{x}_{l}$, $\forall l$. 

\section{Asymptotic PEP analysis and achievable diversity order}
\label{sec:asy}

The PEP has been widely adopted as an efficient metric to quantify the achievable diversity order, which is defined as the slope of the PEP at high signal-to-noise ratio (SNR) regime. The achievable diversity order of the $l$th user can be evaluated as
\begin{equation}
\label{eq:ds}
d_{s} = -\lim_{\bar{\gamma} \rightarrow \infty} \frac{\textup{log}(\textup{Pr}\left ( x_{l}\rightarrow \bar{x}_{l} \right ))}{\textup{log}\bar{\gamma}},
\end{equation}
where $\bar{\gamma}$ is the average transmit SNR of the $l$th user. To obtain the achievable diversity order, we first exploit the asymptotic expansion of the Meijer's G-function from \cite[Eq. (41)]{asymp} to approximate \eqref{eq:pep7} when $\bar{\gamma} \rightarrow \infty$. Note that as $\bar{\gamma} \rightarrow \infty$, the argument of the Meijer's G-function in \eqref{eq:pep7} approaches infinity. Therefore, utilizing the asymptotic expansion of the Meijer's G-function in \cite[Eq. (41)]{asymp}, the PEP of the $l$th user can be asymptotically given as in \eqref{eq:pep_asy}, at the top of next page.  
\begin{figure*}[ht]
\begin{equation}
\label{eq:pep_asy}
\begin{split}
\textup{Pr}\left ( x_{l}\rightarrow \bar{x}_{l} \right ) \sim  a_{1}a_{2}2^{a_{4}+a_{5}-a_{3}} \Bigg [ &-\zeta_{l}^{-\left (\frac{a_4}{2}+1  \right )}\frac{2\Gamma\left ( -\frac{a_{4}-a_{5}}{2} \right )\Gamma\left ( -\frac{a_{4}-a_{5}+1}{2} \right )\Gamma\left ( 1+\frac{a_{4}}{2} \right )}{\Gamma\left ( -\frac{a_{4}-a_{3}}{2} \right )\Gamma\left ( -\frac{a_{4}-a_{3}+1}{2} \right )} \\
& +  \zeta_{l}^{-\left (\frac{a_4}{2}+\frac{1}{2}  \right )}\frac{\Gamma\left ( \frac{1-a_{4}+a_{5}}{2} \right )\Gamma\left ( -\frac{a_{4}-a_{5}}{2} \right )\Gamma\left ( \frac{1+a_{4}}{2} \right )}{\Gamma\left ( \frac{1-a_{4}+a_{3}}{2} \right )\Gamma\left ( -\frac{a_{4}-a_{3}}{2} \right )} \\
& - \zeta_{l}^{-\left (\frac{a_5}{2}+1  \right )}\frac{2\Gamma\left ( -\frac{a_{5}-a_{4}}{2} \right )\Gamma\left ( -\frac{a_{5}-a_{4}+1}{2} \right )\Gamma\left ( 1+\frac{a_{5}}{2} \right )}{\Gamma\left ( -\frac{a_{5}-a_{3}}{2} \right )\Gamma\left ( -\frac{a_{5}-a_{3}+1}{2} \right )} \\
& + \zeta_{l}^{-\left (\frac{a_5}{2}+\frac{1}{2}  \right )}\frac{\Gamma\left ( \frac{1-a_{5}+a_{4}}{2} \right )\Gamma\left ( -\frac{a_{5}-a_{4}}{2} \right )\Gamma\left ( \frac{1+a_{5}}{2} \right )}{\Gamma\left ( \frac{1-a_{5}+a_{3}}{2} \right )\Gamma\left ( -\frac{a_{5}-a_{3}}{2} \right )} \Bigg ]. 
\end{split}
\end{equation} 
\hrulefill
\vspace*{4pt}
\end{figure*}
To evaluate the achievable diversity order, we substitute \eqref{eq:pep_asy} in \eqref{eq:ds} and evaluate the logarithmic function. By considering the dominant components that contribute to the diversity order, the achievable diversity order of the $l$th user can be evaluated as
\begin{equation}
\label{eq:ds2}
\begin{split}
d_{s} &= -\lim_{\bar{\gamma} \rightarrow \infty}\\
&\left (\frac{\textup{log}(\bar{\gamma}^{-\left (\frac{a_4}{2}+1  \right )}+\bar{\gamma}^{-\left (\frac{a_4}{2}+\frac{1}{2}  \right )}+\bar{\gamma}^{-\left (\frac{a_5}{2}+1  \right )}+\bar{\gamma}^{-\left (\frac{a_5}{2}+\frac{1}{2}  \right )})}{\textup{log}(\bar{\gamma})}\right )\\
& = \textup{min}\left ( \frac{a_4}{2}+\frac{1}{2} ,\frac{a_5}{2}+\frac{1}{2} \right ),
\end{split}
\end{equation}
where $a_4$ and $a_5$ are given in \eqref{a4} and \eqref{a5}, respectively. Fig. \ref{fig:a4_a5} shows that $a_5 < a_4$ for all $M$ values. Moreover, it can be noticed from Fig. \ref{fig:a4_a5} that $a_{5}$ is a function of the number of REs, $M$; this indicates that the achievable diversity order of NOMA users is dependent on the number of REs ($M$).  
\begin{figure}[ht] 
\centering
\includegraphics[width=0.9\linewidth]{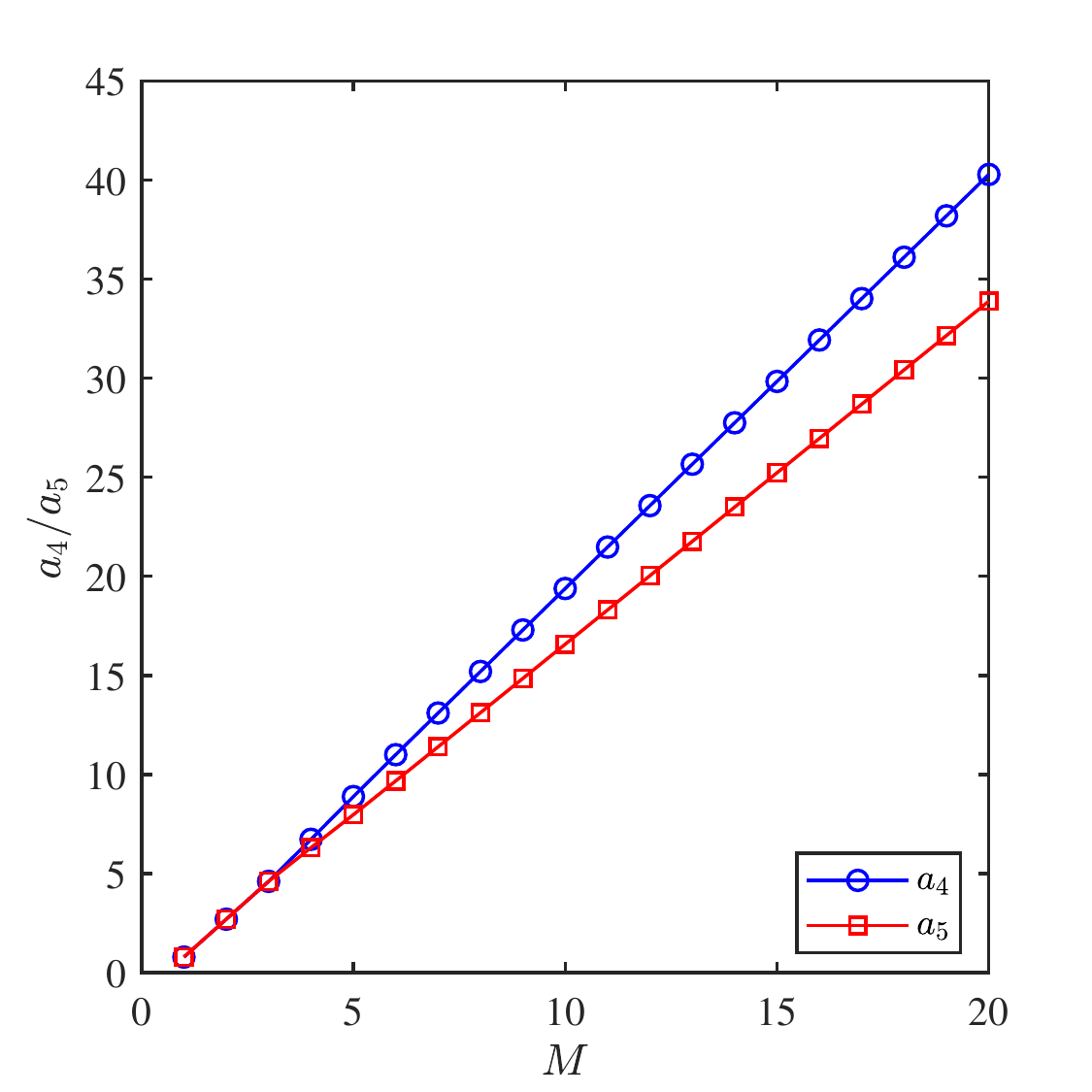}
\caption{Number of REs $M$ vs. $a_4$ and $a_5$.}
\label{fig:a4_a5}
\end{figure}
\begin{figure}[ht] 
\centering
\includegraphics[width=0.9\linewidth]{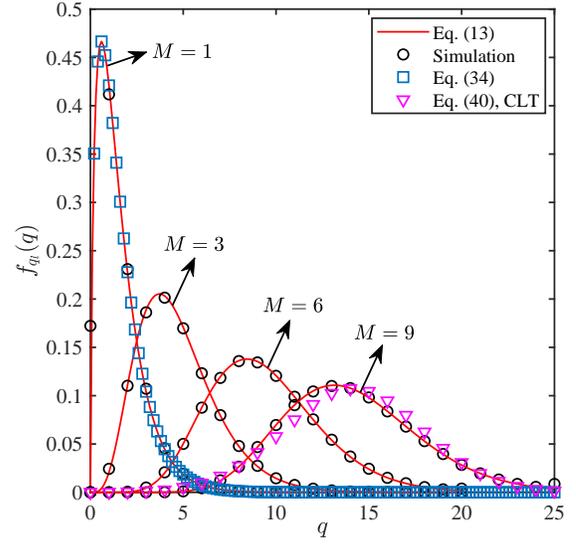}
\caption{Analytical and simulated PDF of $q_{l}$ for different $M$ values.}
\label{fig:pdf}
\end{figure}
\begin{figure}[ht] 
\centering
\includegraphics[width=0.9\linewidth]{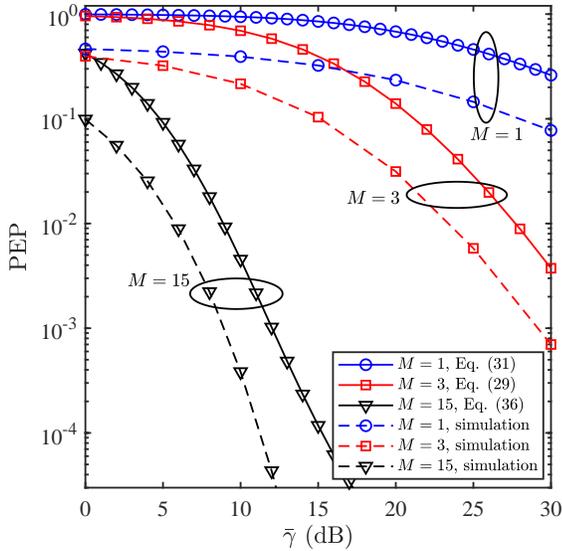}
\caption{PEP of the first NOMA user, $U_{1}$, with different numbers of REs.}
\label{fig:PEP}
\end{figure}
\begin{figure}[ht] 
\centering
\includegraphics[width=0.9\linewidth]{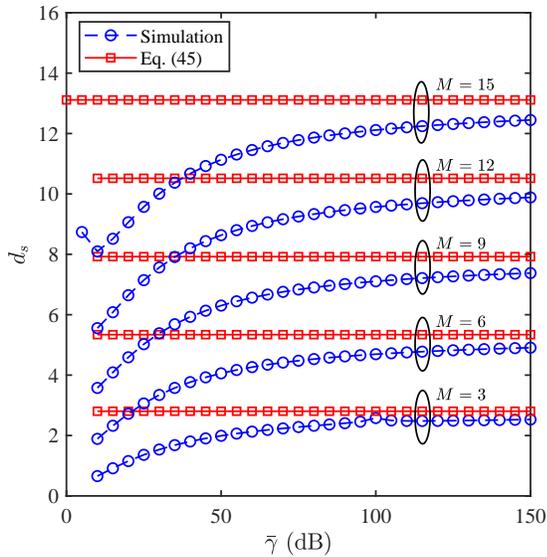}
\caption{Achievable diversity order of the two users for various values of $M$.}
\label{fig:ds}
\end{figure}
\begin{figure}[ht] 
\centering
\includegraphics[width=0.9\linewidth]{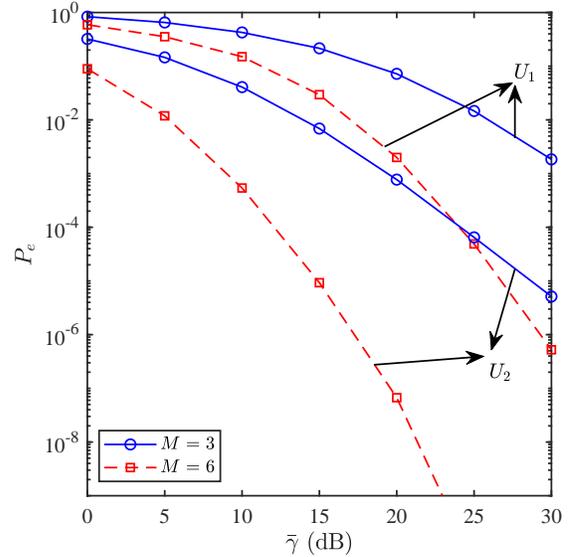}
\caption{BER union bound of the two users versus $\bar{\gamma}$ for different number of REs.}
\label{fig:UB}
\end{figure}
\begin{figure}[ht] 
\centering
\includegraphics[width=0.9\linewidth]{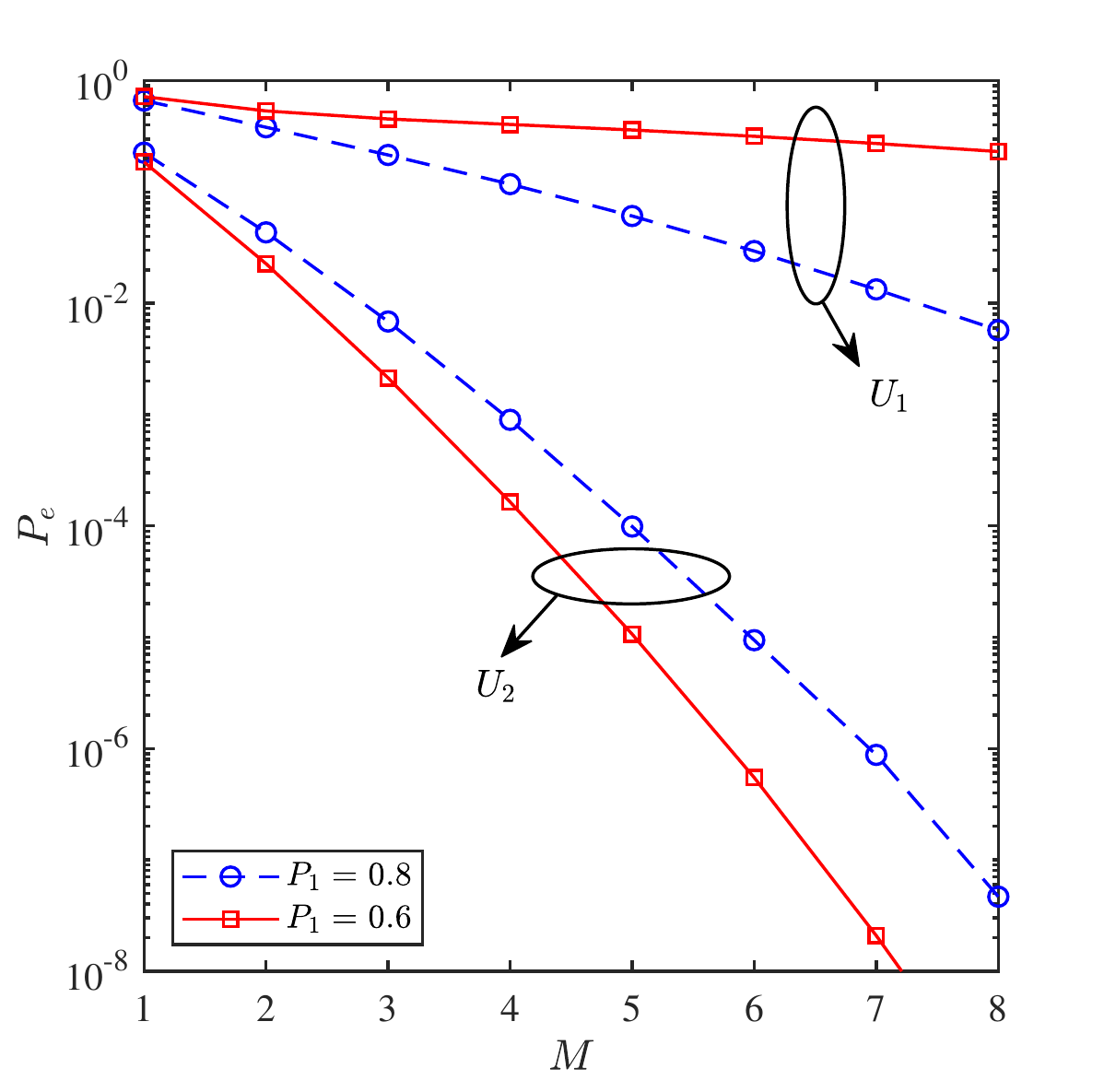}
\caption{BER union bound of the two users versus $M$ for different power allocation coefficients, and $\bar{\gamma}=15$ dB.}
\label{fig:UBN}
\end{figure}

\section{Numerical and Simulation Results}
\label{sec:result}

In this section, we present analytical and Monte Carlo simulation results to corroborate the derived mathematical framework and present insightful conclusions into the performance of LIS-assisted NOMA systems. In particular, we consider a downlink NOMA system with a single BS, a single LIS with $M$ REs and two users, i.e., $L=2$. Note that users clustering is used for $L>3$. Unless mentioned otherwise, we assume that the first user, $U_{1}$, is the far user located at distance $d_{R,1} = 5$ from the LIS, while the second user, $U_{2}$, is the near user located at distance $d_{R,2}=2$ from the LIS. Without loss of generality, the distance between the BS and the LIS is normalized to unity, i.e., $d_{B} =1$. The power coefficients are allocated as $P_{1} = 0.8$ and $P_{2}=0.2$ \cite{Barrieh}, and the transmitted and detected symbols of all users are selected randomly from a binary phase shift keying (BPSK) constellation. Moreover, the path-loss exponent is $\alpha =3$ \cite{Ding}.

Fig. \ref{fig:pdf} shows the analytical and simulated PDF of the e2e fading coefficient, $q_{l}$, for $M=1,3,6,$ and 9. The results in Fig. \ref{fig:pdf} confirm the accuracy of the estimated PDF in Proposition \ref{lemma:pdf}, and one can observe that the analytical and simulation results perfectly match. Additionally, Fig. \ref{fig:pdf} corroborates the accuracy of the derived PDF in Proposition \ref{lemma:pdf} for the special cases of $M=1$ and large $M$. Specifically, it can be noticed that \eqref{ccdf1} perfectly matches \eqref{eq:PDFDR1} for $M=1$. Likewise, for larger values of $M$, it is observed that \eqref{eq:barpdf}, derived based on the CLT, closely matches \eqref{ccdf1}. Finally, it is noticed that as $M$ increases, the mean value of $q_{l}$ increases, resulting in a better average error rate performance.

Fig. \ref{fig:PEP} depicts the PEP performance of the first user, where the analytical and simulation results are presented for $M=1, 3$ and $15$. The simulation results validate the accuracy of the derived expressions in \eqref{eq:pep7}, \eqref{eq:pepm2}, and \eqref{eq:pepm4}, where it can be observed that the analytical results provide a tight upper bound on the PEP performance over the entire SNR range. Specifically, one can notice that both analytical and simulation results experience the same diversity order, which further confirms the validity of the derived mathematical framework.
 
The observation in \eqref{eq:ds2} is further validated in Fig. \ref{fig:ds}, where the achievable diversity order of the first user is presented for $M=3,6,9,12$ and 15. For the sake of clarity, the results of the second user are omitted from Fig. \ref{fig:ds}, as they are identical to those of the first user. The results in Fig. \ref{fig:ds} show that the achievable diversity order in the LIS-based NOMA system converges to $\frac{a_5}{2}+\frac{1}{2}$, where $a_5$ is a function of $M$, as depicted in Fig. \ref{fig:a4_a5}. This observation is verified by Monte Carlo simulations. Fig. \ref{fig:ds} highlights the prominent potential of LIS in offering notable enhancement into the error rate performance of NOMA systems, as it is observed that the achievable diversity order of the underlying system model is primarily dependent on the number of REs, $M$.  

The BER union bound of the underlying LIS-assisted NOMA system with two users scenario is shown in Fig. \ref{fig:UB}, for $M=3$ and $M=6$. The results confirm the advantages of utilizing LIS to enhance the error rate performance of NOMA. Specifically, in addition to the extended coverage, the remarkably enhanced error rate performance can be a key driver behind the integration of LIS into NOMA systems. For example, it can observed from Fig. \ref{fig:UB} that as $M$ changes from 3 to 6, the performances of the first and second users are enhanced by approximately 9.8 dB and 10 dB, respectively, at $P_{e} \approx 1.8 \times 10^{-3}$. Moreover, it can be noticed from Fig. \ref{fig:UB} that both users experience the same diversity order.  

The effect of the power coefficients and the number of REs is further investigated in Fig. \ref{fig:UBN}, where the BER union bound versus $M$ is presented for $P_{1}=0.8$ and 0.6, at $\bar{\gamma}=15$ dB. Note that $P_{1}>P_{2}$ and $P_{1}+P_{2} = 1$. It can be noticed from Fig. \ref{fig:UBN} that the first user is more susceptible to changes in the power allocation, over different $M$ values. This stems from the fact that the first user does not perform SIC, and hence, changes in the power allocation result in variations in the interference level caused by the second user.

\section{Conclusion}
\label{sec:con}

In this paper, we proposed a comprehensive mathematical framework to investigate the error rate performance of an LIS-assisted NOMA system. Specifically, we derived the PDF of the e2e fading channel, which is then utilized to obtain an accurate PEP expression of NOMA users. As special cases, we derived novel PEP expressions for a single and large number of REs, where for the latter we adopted the central limit theorem to obtain an alternative PDF for the e2e fading channel. The derived PEP expression was then exploited to evaluate a tight union bound on the BER. Furthermore, we evaluated the asymptotic PEP of the $l$th user, which was then used to quantify the achievable diversity order. The derived analytical results of the considered scenario, validated by Monte Carlo simulations, showed that the achievable diversity order of LIS-based NOMA systems depends on the number of REs ($M$). Finally, the obtained results highlighted the advantages of utilizing LIS to improve the diversity order of NOMA users, compared to the conventional amplify-and-forward relaying, which limits the achievable diversity order of NOMA users to unity \cite{Barrieh3}.

\appendix
In this appendix, we provide the proof of Proposition \ref{lemma:pdf} and derive the first four moments required for the derivation of the PDF of $q_{l}$.
\subsection{The derivation of the approximate PDF $f_{q_{l}}(q)$}
By recalling that both $\left |h_{m}\right |$ and  $\left |g_{m,l}\right |$, $\forall m=1,\cdots,M$ are independent and identically Rayleigh distributed with zero mean and variance $2 \sigma^{2}$, the PDF of $\beta_{m}^{(l)} = \left |h_{m}\right | \left |g_{m,l}\right |$, which follows the double Rayleigh distribution, is given as the following \cite{DR} \vspace{-0.2cm}
\begin{equation}\vspace{-0.2cm}
f_{\beta_{m}^{(l)}}\left(b\right)=\frac{b}{\sigma^{4}}K_{0}\left(\frac{b}{\sigma^{2}}\right),b>0.\label{eq:PDFDR}
\end{equation} 
By exploiting \cite[Eq. (03.04.26.0008.01)]{bessel}, the PDF of $\beta_{m}^{(l)}$ can be rewritten as \vspace{-0.2cm} 
\begin{align} \vspace{-0.2cm}
f_{\beta_{m}^{(l)}}\left(b\right) & =\frac{b}{2\sigma^{4}}G_{0,2}^{2,0}\left(\frac{b^{2}}{4\sigma^{4}}\left\vert \begin{array}{c}
-;-\\
0,0;-
\end{array}\right.\right) \nonumber \\
 & =\frac{b}{2\sigma^{4}}\frac{1}{2\pi j}\int_{\mathcal{C}}\left(\frac{b^{2}}{4\sigma^{4}}\right)^{-s}\Gamma^{2}\left(s\right)ds\nonumber \\
 & =\frac{1}{\sigma^{2}}\frac{1}{2\pi j}\int_{\mathcal{C}}\left(\frac{b}{2\sigma^{2}}\right)^{-2s+1}\Gamma^{2}\left(s\right)ds,
\end{align}
where $\mathcal{C}$ is an appropriate complex contour ensuring the convergence of the above Mellin-Barnes integral (e.g. $[1/2-j\infty,1/2+j\infty]$), and $j=\sqrt{-1}$.
\begin{figure*}[ht]
\setcounter{equation}{54} 
\begin{align}\label{mu3a}
\mu_{3} & =\mathbb{E}\left[\left(\sum_{i=1}^{M}\beta_{i}^{(l)}\right)^{3}\right]=\left\{ \begin{array}{l}
\underset{\frac{9\pi M\sigma^{6}}{2}}{\underbrace{\sum_{i=1}^{M}\mathbb{E}\left[(\beta_{i}^{(l)})^{3}\right]}}+\underset{6\pi M(M-1)\sigma^{6}}{\underbrace{\underset{M(M-1)}{3\underbrace{\sum_{i=1}^{M}\sum_{\substack{\ell=1\\
\ell\neq i
}
}^{M}}}\underset{\frac{\pi}{2}\sigma^{2}\times4\sigma^{4}}{\underbrace{\mathbb{E}\left[\beta_{i}^{(l)}\right]\mathbb{E}\left[(\beta_{\ell}^{(l)})^{2}\right]}}}}+\underset{M(M-1)(M-2)\left(\frac{\pi}{2}x\right)^{3}\sigma^{6}}{\underbrace{6\underset{\binom{M}{3}}{\underbrace{\sum_{i=1}^{M-2}\sum_{\ell=i+1}^{M-1}\sum_{k=\ell+1}^{M}}}\underset{\left(\frac{\pi}{2}\sigma^{2}\right)^{3}}{\underbrace{\mathbb{E}\left[\beta_{i}^{(l)}\right]\mathbb{E}\left[\beta_{\ell}^{(l)}\right]\mathbb{E}\left[\beta_{k}^{(l)}\right]}}}} \\ \;\;\;\;\;\;\;\;\;\;\;\;\;\;\;\;\;\;\;\;\;\;\;\;\;\;\;\;\;\;\;\;\;\;\;\;\;\;\;\;\;\;\;\;\;\;\;\;\;\;\;\;\;\;\;\;\;\;\;\;\;\;\;\;\;\;\;\;\;\;\;\;\;\;\;\;\;\; \;\;\;\;\;\;\;\;\;\;\;\;\;\;\;\;\;\;\;\;\;\;\;\;\;\;\;\;\;\;\;\;\;\;\;\;\;\;\;\;\;\;\;\;\;\;\;\;\;\;\;\; ,M\geq3 \\
9\pi\sigma^{6}+3\times2\times\frac{\pi}{2}\sigma^{2}\times4\sigma^{4}=21\pi\sigma^{6},M=2\\
\frac{9\pi\sigma^{6}}{2},M=1
\end{array}\right.
\end{align}
\hrulefill
\vspace*{4pt}
\end{figure*}
\begin{figure*}[ht]
\setcounter{equation}{55} 
\begin{equation}\label{mu4a}
\mu_{4}=\mathbb{E}\left[\left(\sum_{i=1}^{M}\beta_{i}^{(l)}\right)^{4}\right]=\left\{ \begin{array}{l}
\underset{64M\sigma^{8}}{\underbrace{\sum_{i=1}^{M}\mathbb{E}\left[(\beta_{i}^{(l)})^{4}\right]}}+\underset{6\times\binom{M}{2}\times\left(4\sigma^{4}\right)^{2}}{\underbrace{6\sum_{i=1}^{M-1}\sum_{j=i+1}^{M}\mathbb{E}\left[(\beta_{i}^{(l)})^{2}(\beta_{j}^{(l)})^{2}\right]}}+\underset{4\times M\times\left(M-1\right)\times\left(\frac{\pi}{2}\sigma^{2}\right)\times\frac{9\pi\sigma^{6}}{2}}{\underbrace{4\sum_{i=1}^{M}\sum_{\underset{j\neq i}{j=1}}^{M}\mathbb{E}\left[(\beta_{i}^{(l)})^{3}\beta_{j}^{(l)}\right]}} \\ \;\;\;\;\;\;\;\;\;\;\;\;\;\;\;\;\;\;\;\;\;\;\;\;\;\; +\underset{12M\times\frac{\left(M-1\right)(M-2)}{2}\left(\frac{\pi}{2}\sigma^{2}\right)^{2}\times4\sigma^{4}}{\underbrace{12\sum_{i=1}^{M}\sum_{\underset{j\neq i}{j=1}}^{M}\sum_{\underset{k\neq i}{k>j}}^{M}\mathbb{E}\left[(\beta_{i}^{(l)})^{2}\beta_{j}^{(l)}\beta_{k}^{(l)}\right]}} +\underset{24\times\binom{M}{4}\times\left(\frac{\pi}{2}\sigma^{2}\right)^{4}}{\underbrace{24\sum_{i<j<k<p}\mathbb{E}\left[\beta_{i}^{(l)}\beta_{j}^{(l)}\beta_{k}^{(l)}\beta_{p}^{(l)}\right]}},M\geq4\\
3\times64\sigma^{8}+\underset{6\times3\times\left(4\sigma^{4}\right)^{2}}{\underbrace{6\sum_{i=1}^{2}\sum_{j=i+1}^{3}\mathbb{E}\left[(\beta_{i}^{(l)})^{2}(\beta_{j}^{(l)})^{2}\right]}}+\underset{4\times6\times\frac{\pi}{2}\sigma^{2}\times\frac{9\pi\sigma^{6}}{2}}{\underbrace{4\sum_{i=1}^{3}\sum_{\underset{j\neq i}{j=1}}^{3}\mathbb{E}\left[(\beta_{i}^{(l)})^{3}\beta_{j}^{(l)}\right]}}\\ \;\;\;\;\;\;\;\;\;\;\;\;\;\;\;\;\;\;\;\;\;\;\;\; +\underset{12\times3\times\left(\frac{\pi}{2}\sigma^{2}\right)^{2}\times4\sigma^{4}}{\underbrace{12\left(\mathbb{E}\left[(\beta_{1}^{(l)})^{2}\beta_{2}^{(l)}\beta_{3}^{(l)}\right]+\mathbb{E}\left[(\beta_{2}^{(l)})^{2}\beta_{1}^{(l)}\beta_{3}^{(l)}\right]+\mathbb{E}\left[(\beta_{3}^{(l)})^{2}\beta_{1}^{(l)}\beta_{2}^{(l)}\right]\right)}},M=3\\
2\times64\sigma^{8}+\underset{4\times2\times\frac{\pi}{2}\sigma^{2}\times\frac{9\pi\sigma^{6}}{2}}{\underbrace{4\left(\mathbb{E}\left[\beta_{1}^{(l)}\right]\mathbb{E}\left[(\beta_{2}^{(l)})^{3}\right]+\mathbb{E}\left[\beta_{2}^{(l)}\right]\mathbb{E}\left[(\beta_{1}^{(l)})^{3}\right]\right)}}+\underset{6\times\left(4\sigma^{4}\right)^{2}}{\underbrace{6\mathbb{E}\left[(\beta_{1}^{(l)})^{2}\right]\mathbb{E}\left[(\beta_{2}^{(l)})^{2}\right]}},M=2\\
64\sigma^{8},M=1
\end{array}\right.
\end{equation}
\hrulefill
\vspace*{4pt}
\end{figure*}
\setcounter{equation}{47} 
Using the change of variable approach, $t=2s-1$, one can see \vspace{-0.2cm}
\begin{align}\vspace{-0.2cm}
f_{\beta_{m}^{(l)}}\left(b\right) & =\frac{1}{2\sigma^{2}}\frac{1}{2\pi j}\int_{\mathcal{C}}\left(\frac{b}{2\sigma^{2}}\right)^{-t}\Gamma^{2}\left(\frac{t+1}{2}\right)dt\nonumber \\
 & =\frac{1}{2\sigma^{2}}H_{0,2}^{2,0}\left(\frac{b}{2\sigma^{2}}\left\vert \begin{array}{c}
-;-\\
\left(\frac{1}{2},\frac{1}{2}\right),\left(\frac{1}{2},\frac{1}{2}\right);-
\end{array}\right.\right),
\end{align}
where $H^{m,n}_{p,q}\left ( .\mid. \right )$ denotes the Fox's H-function \cite{Fox}. By recalling that the summation of $M$ H-distributions can be tightly approximated by an H-distribution \cite{Faissal}, the PDF of $q_{l} = \sum_{m=1}^{M}\beta_{m}^{(l)}$ can be accurately approximated by a Fox's H-function, relying on the moment-based density approximants method \cite{Faissal}. The resultant Fox's H-function can be represented as a Meijer's G-function, as depicted in Proposition \ref{lemma:pdf}. The tightness of the approximation relies on the number of considered moments. In this work, we use the first four moments, (i.e., $i=1,2,3,4$), which provides a linear system with four equations and four unknown variables ($\mu_i, i \leq 4$). In the following, we provide detailed steps for the derivation of these moments.\vspace{-0.2cm}
\subsection{Moments derivation, $\mu_{i}$, for $1 \leq i \leq 4$}
To obtain the moments of $q_{l}$, we initially derive the first four moments of $\beta_{m}^{(l)}$, for $1 \leq m \leq 4$, as follows \vspace{-0.2cm}
\begin{align}\label{m1} \vspace{-0.2cm}
\mu_{1}^{(m)} & =\mathbb{E}\left[\left|h_{m}\right|\left|g_{m,l}\right|\right]\nonumber \\
 & =\left(\sigma\sqrt{\frac{\pi}{2}}\right)^{2}=\frac{\pi}{2}\sigma^{2},
\end{align}
\begin{align} \vspace{-0.2cm}
\mu_{2}^{(m)} & =\mathbb{E}\left[\left|h_{m}\right|^{2}\left|g_{m,l}\right|^{2}\right]\nonumber \\
 & =\left(2\sigma^{2}\right)^{2}=4\sigma^{4},
\end{align}
\begin{align} \vspace{-0.2cm}
\mu_{3}^{(m)} & =\mathbb{E}\left[\left|h_{m}\right|^{3}\left|g_{m,l}\right|^{3}\right]\nonumber \\
 & =\left(3\sigma^{3}\sqrt{\frac{\pi}{2}}\right)^{2}=\frac{9\pi\sigma^{6}}{2},
\end{align}
and \vspace{-0.2cm}
\begin{align}\label{m4}
\mu_{4}^{(m)} & =\mathbb{E}\left[\left|h_{m}\right|^{4}\left|g_{m,l}\right|^{4}\right]\nonumber \\
 & =\left(8\sigma^{4}\right)^{2}=64\sigma^{8}.\vspace{-0.2cm}
\end{align}
Using \eqref{m1}-\eqref{m4} and employing the multinomial theorem, the first four moments of $q_{l}$ can be evaluated as \vspace{-0.2cm}
\begin{align} \vspace{-0.2cm}
\mu_{1} & =\sum_{m=1}^{M}\mu_{1}^{(m)}=\frac{M\pi\sigma^{2}}{2}
\end{align}
and
\begin{align}\label{eq:mu2}\vspace{-0.2cm}
\mu_{2} & =\mathbb{E}\left[\left(\sum_{m=1}^{M}\beta_{m}^{(l)}\right)^{2}\right]\nonumber \\
 & =\underset{4M\sigma^{4}}{\underbrace{\sum_{m=1}^{M}\mathbb{E}\left[(\beta_{m}^{(l)})^{2}\right]}}+2\sum_{m=1}^{M-1}\sum_{j=m+1}^{M}\underset{\left(\frac{\pi}{2}\sigma^{2}\right)^{2}}{\underbrace{\mathbb{E}\left[\beta_{m}^{(l)}\right]\mathbb{E}\left[\beta_{j}^{(l)}\right]}}\nonumber \\
 & =4M\sigma^{4}+M(M-1)\sigma^{4}\frac{\pi^{2}}{4}.
\end{align} 
Also, $\mu_{3}$ and $\mu_{4}$ are provided in \eqref{mu3a} and \eqref{mu4a}, respectively, given at the top of next page. By performing some algebraic operations, \eqref{eq:mu2}-\eqref{mu4a} can be reduced as in \eqref{m2a}-\eqref{mu4}, respectively. This concludes the proof of Proposition 1.

\bibliographystyle{IEEEtran}
\bibliography{references}

\balance

\end{document}